\documentclass[a4paper,11pt]{article}

\usepackage{hyperref}
\usepackage[a4paper]{geometry}
\usepackage{amsmath}
\usepackage[T1]{fontenc}
\usepackage{amsmath}
\usepackage{amssymb}
\usepackage{textcomp}
\usepackage{amsthm}
\usepackage{amsfonts}
\usepackage{mathrsfs}
\usepackage{array}
\usepackage[english]{babel}
\usepackage{lscape}
\usepackage{graphicx}
\usepackage{graphics}
\usepackage[all]{xy}
\newtheorem{theorem}{Theorem}[section]
\newtheorem{definition}{Definition}[section]
\newtheorem{corollary}{Corollary}[section]
\newtheorem{rmq}{Remark}[section]

\newcommand{\ord}{\texttt{ord}}

\numberwithin{equation}{section}

\title{Vectorial FCSR constructed on totally ramified extension of the $p$-adic numbers}
\author{Abdelaziz Marjane,\\
LAGA, UMR CNRS 7539, Universit\'e Paris 13, Villetaneuse, France}
\date{}
\begin{document}
\maketitle
\begin{abstract}In this paper, we introduce a vectorial conception of $d$-
FCSRs to build these registers over any finite field. We describe the
structure of $d$-vectorial FCSRs and we develop an analysis to obtain
basic properties like periodicity and the existence of maximal length
sequences. To illustrate these  vectorial $d$-FCSRs, we present simple examples and we
compare with those of Goresky, Klapper and Xu.\\
\textbf{Keywords}: LFSR, FCSR, vectorial FCSR, d-FCSR, sequences, periodicity, $p$-adic, $\pi$-adic, maximal period.

\end{abstract}
\section{Introduction}

Linear Feedback Shift Register (LFSR) sequences are used in many applications in Cryptography and Telecommuncations (see Figure \ref{LFSR}. In fact, most of Pseudo-Randomn Generators are based on LFSR. To study LFSR sequences, we use the ring of the power series with coefficients in the finite fields $\mathbb{F}_{p^{n}}$, denoted by  $\mathbb{F}_{p^{n}}[X]$. An output sequence $\underline{a}=(a_{0},a_{1},\cdots,a_{i},\cdots)$ is associated to the power serie $a(X)=a_{0}+a_{1}X+\cdots+a_{i}X^{i}+\cdots$ and we find that $a(X)$ is a rational fraction in $\mathbb{F}_{p^{n}}(X)$ of the form $\displaystyle\frac{s(X)}{q(X)}$, where $q(X)$ is a polynomial defined by the LFSR and called the connection polynomial. So, we obtain the basic properties of $\underline{a}$ like periodicity and distributional properties. For example, the period of $\underline{a}$ divides the order of $X$ modulo $q(X)$.
\begin{figure}[!h]
\begin{center}
\includegraphics[height=3cm,width=7cm]{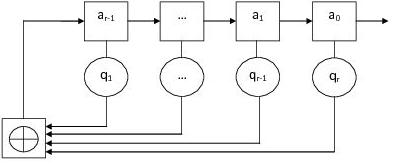}
\end{center}
\caption{Linear Feedback Shift Register or LFSR.}
\label{LFSR}
\end{figure}

Feedback with Carry Shift Registers are a class of non-linear generators and were first introduced by Goresky and Klapper in 1993 (see \cite{Goresky1993,Goresky1997}). An FCSR is constructed like an LFSR but we add a memory cell or a "carry" (see Figure \ref{FCSR}). To analyse FCSR sequences, we use the ring of the $p$-adic integers, denoted by $\mathbb{Z}_{p}$. An output sequence is associated to its $p$-adic expansion $\alpha=a_{0}+a_{1}p+\cdots+a_{i}p^{i}+\cdots$ and we find that $\alpha$ is a rational number in $\mathbb{Q}$ of the form $\displaystyle\frac{s}{q}$ where $q=-1+q_{1}p+\cdots+q_{r}p^{r}$ is an integer defined by the connection coefficients of the register and called the connection integer. So, we obtain the basic properties of FCSR sequences like periodicity, distributional properties, existence of maximal length sequences. As for LFSRs, the period of $\underline{a}$ divides the order of $p$ modulo $q$.
\begin{figure}[!h]
\begin{center}
\includegraphics[height=3cm,width=7cm]{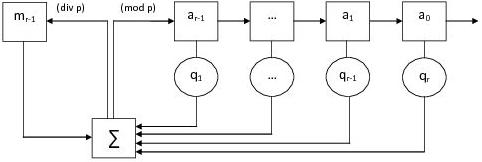}
\end{center}
\caption{Feedback with Carry Shift Register or FCSR.}
\label{FCSR}
\end{figure}

In 1994, Goresky and Klapper introduced $d$-FCSRs (see \cite{Goresky1994}). These registers are like FCSRs but we must add $d$ memories and a "jump" in the carry register (see Figure \ref{d-FCSR}). These results in the analysis through the use of the ring of the $\pi$-adic integers, denoted by $\mathbb{Z}_{p}[\pi]$ (here $\pi=\sqrt[d]{p}$); or the totally ramified extension of $\mathbb{Z}_{p}$, i.e. the ramification index is $d$. An output sequence is associated to the $\pi$-adic integer $\alpha=a_{0}+a_{1}\pi+\cdots+a_{i}\pi^{i}+\cdots$ and we find that $\alpha$ is a fraction in $\mathbb{Q}(\pi)$ of the form $\displaystyle\frac{s}{q}$ where $q=-1+q_{1}\pi+\dots+q_{r}\pi^{r}$ is an element in $\mathbb{Z}[\pi]$ called the connection integer. So Goresky and Klapper obtain some properties of $d$-FCSR sequences. Under some conditions, the period is maximal, i.e. it is the order of $p$ modulo $|\mathtt{N}(q)|$, where $\mathtt{N}(q)$ is the norm of $q$ in $\mathbb{Q}(\pi)$ as a $\mathbb{Q}$-vector space of dimension $d$.
\begin{figure}[!h]
\begin{center}
\includegraphics[height=3cm,width=9cm]{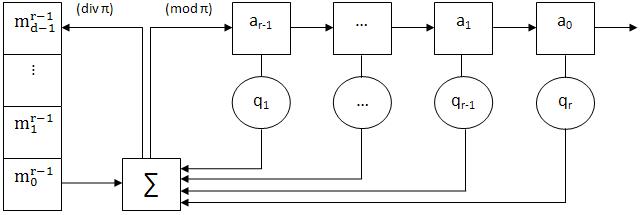}
\end{center}
\caption{$d$-FCSR.}
\label{d-FCSR}
\end{figure}

To extend the construction of FCSR to any finite fields $\mathbb{F}_{p^{n}}$, the authors introduced the notion of a vectorial conception in 2010 (see \cite{marjane2010,allailou2010}). In fact, FCSR and $d$-FCSR are mainly developped on $\mathbb{F}_{p}$ (where $p$ is prime). The analysis of vectorial FCSR is based on $\mathbb{Z}_{p^{n}}$ the absolutely unramified extension of $\mathbb{Z}_{p}$. The fields of $p^{n}$ elements can be constructed as the quotient ring $\mathbb{F}_{p}[X]/(P)$ where $P(X)$ is a primitive polynomial over $\mathbb{F}_{p}$. Consider a root of $P$ denoted by $\beta$, the field $\mathbb{F}_{p^{n}}$ is the smallest field extension of $\mathbb{F}_{p}$ and containing $\beta$, denoted by $\mathbb{F}_{p}[\beta]$. On the other hand, $\mathbb{Z}_{p^{n}}$ can be constructed as a quotient ring $\mathbb{Z}_{p}[X]/(P)$. An output sequence $\underline{a}$ of Vectorial FCSR (see Figure \ref{VFCSR}) is first decomposed as a vector of $n$ $p$-ary sequences denoted by $(\underline{a}_{0},\cdots,\underline{a}_{n-1})$. Each component $\underline{a}_{j}$ is associated to its $p$-adic expansion $\alpha_{j}$. So $\underline{a}$ is associated to a vector $\alpha=(\alpha_{0},\cdots,\alpha_{n-1})$ in $(\mathbb{Z}_{p})^{n}$ and we find that $\alpha$ is a rational vector in $\mathbb{Q}^{n}$ of the form $\displaystyle\frac{1}{\mathtt{N}(q)}\mathbb{Z}^{n}$. Recall that $\mathtt{N}(q)$ is the norm of $q$ in $\mathbb{Q}[X]/(P)$ as a $\mathbb{Q}$-vector space of dimension $n$. The connection integer $q$ is an element of $\mathbb{Z}[\beta]$. Then, we obtain the basic properties of VFCSR sequences: all VFCSR sequences are eventually periodic and its period divides the order of $p$ modulo $|\mathtt{N}(q)|$.\\
\begin{figure}[!h]
\begin{center}
\includegraphics[height=7cm,width=9cm]{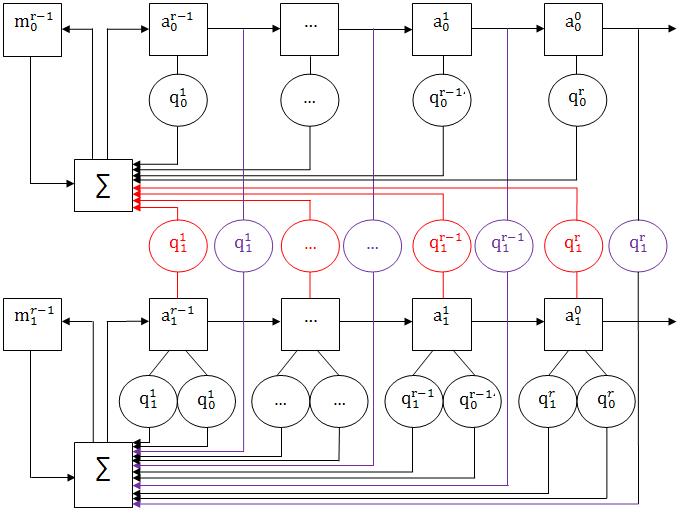}
\caption{Vectorial FCSR.}
\label{VFCSR}
\end{center}
\end{figure}

In 1999, Klapper and Xu have presented a generalization of FCSR and LFSR called AFSR (see \cite{klapper1999}). Algebraic FSR are constructed over an integral and commutative
ring $R$ and they consider an element $\pi \in R$. We assume that $R/\pi R$
is a finite field. The construction is the same as that of FCSR. Also, if we
take $R=\mathbb{F}_{2^{n}}[x]$ and $\pi =x$ where $x$ is an indeterminate,
we obtain an LFSR. If we take $R=\mathbb{Z}$ and $\pi =p$ with $p$
prime, you have $p$-ary FCSR. If we take $R=\mathbb{Z}[\sqrt[d]{2}]$ and $\pi
=\sqrt[d]{2}$, we have the case of $d$-FCSR. This construction generalizes
all FSR over an algebraic structure. Any choice of $\pi $ defines a new topology on $%
R $. They consider the completion of $R$ for the $\pi $-adic
topology. If $R$ is noetherian, this completion is simply the set of power
series $a_{0}+a_{1}\pi+\cdots+a_{i}\pi ^{i}+\cdots$. For analysis, we have a
correspondance between this set of power series and the set of sequences
over $R/\pi R$. Furthermore, the most important results are that the output
sequence is the $\pi $-adic expansion of $\displaystyle\frac{p}{q}$ which is in the
fraction field of $R$ and $q=-1+q_{1}\pi+\cdots+q_{r}\pi^{r}$ is the
connection integer. Under more assumptions, the period is the order of $\pi $ modulo $q$. All these results about algebraic FSR and their properties are clearly described in a book \cite{goresky2009}.

In this paper, we extend the construction of $d$-FCSR to any finite fields $\mathbb{F}_{p^{n}}$. We adapt the vectorial conception to $d$-FCSR and we call these new registers $d$-Vectorial FCSR. The main idea is to consider $P(X)$ a primitive polynomial on $\mathbb{F}_{p}$ and $\overline{\beta}$ a root of $P$. The fields $\mathbb{F}_{p^{n}}$ is isomorphic to $\mathbb{F}_{p}\big[\overline{\beta}\big]$. We decompose an output sequence $\underline{a}$ to $n$ sequences $(\underline{a}_{0},\cdots,\underline{a}_{n-1})$ on $\mathbb{F}_{p}$ and for all $0\leq j\leq n-1$, we consider $\underline{a}_{k,j}=(a^{j}_{k},a^{j}_{k+d},a^{j}_{k+2d},\cdots)$ the $d$-decimations of the $k$-shifts (here $0\leq k\leq d-1$) of $\underline{a}_{j}$ and their $p$-adic expansion. We obtain a collection of $nd$ sequences on $\mathbb{F}_{p}$ and $\alpha^{'}$ a vector in $(\mathbb{Z}_{p})^{nd}$. The first main result is:
\begin{theorem}
The $p$-adic vector $\alpha^{'}$ is a rational vector in $\displaystyle\frac{1}{\mathtt{N}(q)}\mathbb{Z}^{nd}$, where $q$ is the connection integer in $\mathbb{Z}[\pi,\beta]$ and $\mathtt{N}$ is the norm in $\mathbb{Q}(\pi,\beta)$ as a $\mathbb{Q}$-vector-space of dimension $nd$.
\end{theorem} 
So for all $k=0,\cdots,d-1$ and $n=0,\cdots,n-1$, the "sub-sequences" $\underline{a}_{k,j}$ are eventually periodics, then $\underline{a}$ is eventually periodic (see Theorem \ref{thm1}, Section 4 and 5). In the following, we describe this analysis in detail. Note that in \cite{klapper1999} (p.19-25), Klapper and Xu give three examples to describe the different cases of AFSR and one of them corresponds to $d$-VFCSR. In fact, AFSR generalize LFSR, $p$-ary FCSR, $p$-ary $d$-FCSR, VFCSR and $d$-VFCSR, but the analysis of AFSR is very formal and the period of an output sequence of AFSR divides $\ord_{q}(\pi)$ and the maximal period is $\ord_{q}(\pi)$. In this paper, we give a practical and easily implementable description of $d$-VFCSRs and an analysis reduced to $p$-adic framework and specific to these registers. We show also the second main result: the period divides $d.\ord_{\mathtt{N}(q)}(p)$ and the maximal period is this number (see Theorem \eqref{thm2}). It is more easy to compute the order of an integer modulo another integer. Furthermore, to find maximal period sequences, we have to find prime
numbers represented by the $nd$-form defined by $\mathtt{N}(q)$.

The paper is organized as follows : in Section 2, we set the needed algebraic
background. We introduce $d$-VFCSR in Section 3, the analysis is given in Section 4 and we
study the periodicity in Section 5. Finally, we illustrate the $d$-VFCSRs by
three examples in Section 6.

\section{Preliminary Algebraic}

Let $p$ be a prime number, $d,n$ two non negative integers and $\pi=\sqrt[d]{p}$. The polynomial $X^{d}-p$ is irreducible over $\mathbb{Q}$ by the criterion of Eisenstein. Since it is a monic polynomial and $\mathbb{Z}$ is a factorial ring, then $X^{d}-p$ is irreducible over $\mathbb{Z}$. $\mathbb{Q}[X]$ is an euclidean ring because $\mathbb{Q}$ is a field. The surjective homomorphism of rings defined by $\mathbb{Q}[X]\rightarrow \mathbb{Q}[\pi],R(X)\mapsto R(\pi)$ induces an isomorphism of rings between $\mathbb{Q}[X]/(X^{d}-p)$ and $\mathbb{Q}[\pi]$. $X^{d}-p$ is irreducible in the euclidean ring $\mathbb{Q}[X]$, so the ideal $(X^{d}-p)$ generated by $X^{d}-p$ is maximal and $\mathbb{Q}[\pi]$ is a field. It is also the smallest field extension of $\mathbb{Q}$ containing $\pi$ denoted by $\mathbb{Q}(\pi)$. It is a $\mathbb{Q}$-vector space of dimension $d$ with the canonical basis $\left\{1,\pi,\cdots,\pi^{d-1}\right\}$. 

The surjective homomorphism of rings defined by $\mathbb{Z}[X]\rightarrow \mathbb{Z}[\pi],$ $R(X)\mapsto R(\pi)$ induces an isomorphism between $\mathbb{Z}[X]/(X^{d}-p)$ and $\mathbb{Z}[\pi]$. Here $\mathbb{Z}[\pi]$ is a free $\mathbb{Z}$-module of rank $n$ with the canonical basis $\left\{1,\pi,\cdots,\pi^{d-1}\right\}$. The fraction field of $\mathbb{Z}[\pi]$ is $\mathbb{Q}[\pi]$ which is an algebraic number field of degree $d$ and $\mathbb{Z}[\pi]$ is an order of $\mathbb{Q}[\pi]$. All elements of $\mathbb{Z}[\pi]$ written in the form $q_{0}+q_{1}\pi+\cdots+q_{d-1}\pi^{d-1}$, where $q_{i}\in \mathbb{Z}.$
The inclusion $\mathbb{Z}\hookrightarrow \mathbb{Z}[\pi]$ induces an isomorphism of rings between $\mathbb{Z}/p\mathbb{Z}$ and $\mathbb{Z}[\pi]/\pi\mathbb{Z}[\pi]$. $\mathbb{Z}/p\mathbb{Z}$ is a finite field, then $(\pi)$ is maximal (and prime).

The completion of $\mathbb{Q}$ for the $p$-adic valuation is $\mathbb{Q}_{p}$ the field of $p$-adic numbers. Its valuation ring is $\mathbb{Z}_{p}$ the ring of the $p$-adic integers. All elements of $\mathbb{Z}_{p}$ can be developped as power series of the form $\underset{i\in \mathbb{N}}{\sum} a_{i}p^{i}$, where $a_{i}\in \mathbb{F}_{p}$. Up to isomorphism, the ring $\mathbb{Z}_{p}$ is the unique complete discrete valuation ring of the residue field $\mathbb{F}_{p}$.

The completion of $\mathbb{Q}(\pi)$ for the $\pi$-adic valuation is $\mathbb{Q}_{p}(\pi)$ which is a $\mathbb{Q}_{p}$-vector space of dimension $d$ and a totally ramified extension of $\mathbb{Q}$. Its valuation ring is $\mathbb{Z}_{p}[\pi]$ which is a free $\mathbb{Z}_{p}$-module of rank $d$. All elements of $\mathbb{Z}_{p}[\pi]$ can be developped as power series of the form $\underset{i\in \mathbb{N}}{\sum}a_{i}\pi^{i}$. The residue field of $\mathbb{Z}_{p}[\pi]$ is $\mathbb{F}_{p}$.

Consider $\mathbb{F}_{p^{n}}$ the field of $p^{n}$ elements. $\mathbb{F}_{p^{n}}$ is isomorphic to the quotient field $\mathbb{F}_{p}[X]/(\overline{P})$ where $\overline{P}$ is a primitive polynomial of degree $n$ with coefficients in $\mathbb{F}_{p}$. Without loss of generality, we take $\overline{P}$ of the form $X^{n}-\cdots-1$. Let $P$ be the canonical lift of $\overline{P}$ in $\mathbb{Z}[X]$. We identify $P$ to $\overline{P}$. Consider $\beta\in \mathbb{C}$ a root of $P$ and $\overline{\beta}$ its reduction modulo $p$. The field $\mathbb{F}_{p^{n}}$ is isomorphic to $\mathbb{F}_{p}[\overline{\beta}]$.

Since $\mathbb{Z}$ is factorial, $p\mathbb{Z}$ is prime ideal and $P$ is reducible in $\mathbb{F}_{p}[X]$, then $P$ is irreducible in $\mathbb{Q}[X]$. As $P$ is a monic polynomial, then $P$ is also irreducible in $\mathbb{Z}[X]$. By the same arguments of the precedent paragraph, we have an isomorphism of field $\mathbb{Q}[X]/(P)\cong \mathbb{Q}[\beta]$ and an isomorphism of ring $\mathbb{Z}[X]/(P)\cong \mathbb{Z}[\beta]$. Recall that $\mathbb{Q}[\beta]$ is a $\mathbb{Q}$-vector space of dimension $n$ with the canonical basis $\left\{1,\beta,\cdots,\beta^{n-1}\right\}$ and $\mathbb{Z}[\beta]$ is a free $\mathbb{Z}$-module of rank $n$ and an order of the number field $\mathbb{Q}[\beta]$. All elements of $\mathbb{Z}[\beta]$ written in the form $a_{0}+a_{1}\beta+\cdots+a_{n-1}\beta^{n-1}$, where $a_{i}\in \mathbb{Z}$.

The completion of $\mathbb{Q}(\beta)$ for the $p$-adic valuation is $\mathbb{Q}_{p^{n}}=\mathbb{Q}_{p}[X]/(P)$. The field extension $\mathbb{Q}\supseteq\mathbb{Q}(\beta)$ is of degree $n$ and $\mathbb{Q}_{p}[X]/(P)$ is an absolutely unramified extension of $\mathbb{Q}_{p}$. The valuation ring of $\mathbb{Q}_{p^{n}}$ is $\mathbb{Z}_{p^{n}}=\mathbb{Z}_{p}[X]/(P)$. Up to isomorphism, this ring is the unique complete discrete valuation ring of the residue field $\mathbb{F}_{p^{n}}$ and a free $\mathbb{Z}_{p}$-module of rank $n$. Consider the following extensions of fields and rings:
$$\xymatrix{ &\mathbb{Q}(\pi)\ar[rd]^{n^{''}}&&&\mathbb{Z}[\pi]\ar[rd]^{n^{'}}&&\\
\mathbb{Q}\ar[ru]^{d}\ar[rd]^{n}&& \mathbb{Q}(\beta,\pi),&\mathbb{Z}\ar[ru]^{d}\ar[rd]^{n}&&\mathbb{Z}[\pi,\beta]\\
&\mathbb{Q}(\beta)\ar[ru]&&&\mathbb{Z}[\beta]\ar[ru]&&}$$
and $\mathbb{F}_{p}\hookrightarrow\mathbb{Z}[\pi,\beta]/(\pi)$, where $\mathbb{Z}[\pi,\beta]/(\pi)$ is the set of the linear combinations $\underset{0\leq j\leq n-1}{\sum}\overline{a_{j}}(\beta \, \mathtt{mod} \, \pi)^{j}$ where $\overline{a_{j}}\in \mathbb{F}_{p}$. Here $\mathbb{Z}[\pi,\beta]/(\pi)$ is the quotient of the ring $\mathbb{F}_{p}[X]$ by the minimal polynomial of $\beta \, \mathtt{mod} \, \pi$. Note that $P(\beta)=0$ implies $(P \, \mathtt{mod} \, \pi)(\beta \, \mathtt{mod} \, \pi)=0$. Since $P \, \mathtt{mod} \, \pi= \overline{P}$, then $\overline{P}(\beta \, \mathtt{mod} \, \pi)=0$ and the minimal polynomial of $\beta \, \mathtt{mod} \, \pi$ divides $\overline{P}$ in $\mathbb{F}_{p}[X]$. As $\overline{P}$ is primitive, then $\overline{P}$ is the minimal polynomial of $\beta \, \mathtt{mod} \, \pi$. Then $\mathbb{Z}[\pi,\beta]/(\pi) \cong \mathbb{F}_{p^{n}}$.

Since $P(\beta)=0$, then $n^{'}\leq n$. $n^{''}$ is the degree of the extension $\mathbb{Q}(\pi,\beta)\supseteq \mathbb{Q}(\pi)$, then it is the degree of the minimal polynomial of $\beta$ over $\mathbb{Q}(\pi)$ denoted by $P_{0}$. There exists also $P_{1}$ in $\mathbb{Z}[\pi,X]$ of degree $n^{' }$ such that $P_{1}(\beta)=0$. As $Q(\pi,X)$ is euclidean, $P_{0}$ divides $P_{1}$, then $n^{''}\leq n^{'}$. If we multiply $P_{0}$ by $\lambda$ an element of $\mathbb{Z}[\pi]$ such that $\lambda P_{0}$ is in $\mathbb{Z}[\pi,X$] and if we reduce $\lambda P_{0}$ modulo $\pi$, we obtain an anihilator polynomial of $\beta \, \mathtt{mod} \, \pi$ of degree $k\leq n^{''}$ with coefficients in $\mathbb{F}_{p}$. Since $\overline{P}$ is the minimal polynomial of $\beta \, \mathtt{mod} \,\pi$ in $\mathbb{F}_{p}$, then $\overline{P}$ divides $(\lambda P_{0}) \, \mathtt{mod} \, \pi$ and $n \leq k$. Then we have $n=n^{'}=n^{''}$ and we have the following relation: $\big[\mathbb{Q}[\pi,\beta]:\mathbb{Q}[\pi]\big]=\big[\mathbb{Z}[\pi,\beta]:\mathbb{Z}[\pi]\big]=\big[\mathbb{Z}[\pi,\beta]/(\pi):\mathbb{F}_{p}\big]=n$.

$\mathbb{Z}[\pi,\beta]$ is a free $\mathbb{Z}$-module of rank $nd$ with the canonical basis\\ $\left\{1,\cdots,\pi^{i}\beta^{j},\cdots,\pi^{d-1}\beta^{n-1}\right\}$ denoted by $\mathcal{B}$. All elements of $\mathbb{Z}[\pi,\beta]$ written as $\sigma=\sigma_{0,0}+\cdots+\sigma_{i,j}\pi^{i}\beta^{j}+\cdots+\sigma_{d-1,n-1}\pi^{d-1}\beta^{n-1}$, where $\sigma_{i,j}\in \mathbb{Z}$.
The reduction modulo $\pi$ of $\sigma$ is $\sigma\, \mathtt{mod}\,\pi=\underset{j=0}{\overset{n-1}{\sum}}(\sigma_{0,j}\, \mathtt{mod}\,p)(\beta \, \mathtt{mod}\, \pi)^{j}$.  

The completion of $\mathbb{Q}(\pi,\beta)$ for the $\pi$-adic valuation is $\mathbb{Q}_{p^{n}}(\pi)$ the smallest extension field of $\mathbb{Q}_{p}[X]/(P)$ containing $\pi$ and is a totally ramified extension of $\mathbb{Q}_{p^{n}}$. Its valuation ring $\mathbb{Z}_{p^{n}}[\pi]$ is a free $\mathbb{Z}_{p}$-module of rank $nd$. The residue field of $\mathbb{Z}_{p^{n}}[\pi]$ is $\mathbb{F}_{p^{n}}$.
$$\xymatrix{ &\mathbb{Q}_{p}(\pi)\ar[rd]&&&&\mathbb{Z}_{p}[\pi]\ar[rd]&&\\
\mathbb{Q}_{p}\ar[ru]^{d}\ar[rd]^{n}&& \mathbb{Q}_{p^{n}}(\pi)&&\mathbb{Z}_{p}\ar[ru]^{d}\ar[rd]^{n}&&\mathbb{Z}_{p^{n}}[\pi]&\\
&\mathbb{Q}_{p^{n}}\ar[ru]&&&&\mathbb{Z}_{p^{n}}\ar[ru]&&}$$

Consider an extension field $K\subseteq L$. The norm of an element $q\in L$ over $K$ is the determinant of the linear transformation defined by the multiplication by $q$ in $L$ as a $K$-vector space. The norm is denoted by $\mathtt{N}_{K}^{L}(q)$. If $K\subseteq F\subseteq L$,  $\mathtt{N}_{K}^{L}(q)=\mathtt{N}_{K}^{F}(q)\mathtt{N}_{F}^{L}(q)$. In the following, we will consider the norm of $q$ an element of $\mathbb{Q}(\pi,\beta)$ over $\mathbb{Q}$. With respect to the basis $\left\{1,\pi,\cdots,\pi^{d-1}\right\}$, the matrix of the  multiplication by $q=\underset{k=0}{\overset{k=d-1}{\sum}}x_{k}\pi^{k}$ is given by the following matrix:
\begin{equation}
\left(
\begin{array}{ccccccc}
x_{0}  &  px_{d-1} &  \cdots &  px_{2} &  px_{1}\\
x_{1} &  x_{0}& \cdots & px_{3} & px_{2}\\
\vdots &  \vdots       &  \cdots &  \vdots&\vdots\\
x_{d-1} &   x_{d-2}     &         & x_{1}   &  x_{0}
\end{array}
\right).
\label{E}
\end{equation}
\section{Definition of $d$-VFCSR}
\subsection{Formalism}

To construct $d$-VFCSR, we keep the scheme of the binary $d$-FCSR built by Goresky and Klapper and we applicate the vectorial conception. 
Place ourself in the precedent algebraic background to redefine the space in which we calculate. Let $S=\left\{0,\pm 1,\cdots,\pm (p-1)\right\}$. We define the $d$-Vectorial FCSR by the following way:
\begin{definition}A $d$-vectorial feedback with carry shift register over $(\mathbb{F}_{p},P,\mathcal{B})$ of length $r$ with coefficient $q_{1},\cdots ,q_{r}\in S[\beta]$ is an automata or sequence generator whose state is an element $%
s=(a_{0},\cdots ,a_{r-1},m_{r-1})$ where $a_{i}\in \mathbb{F}_{p}\big[ \overline{\beta}\big]$
and $m_{r-1}\in \mathbb{Z}[\pi,\beta]$. We take the canonical lift of the collection of $a_{i}$ in $\mathbb{Z}%
[\beta]$, compute the element $\sigma_{r}=q_{1}a_{r-1}+\cdots+q_{i}a_{r-i}+\cdots+q_{r}a_{0}+m_{r-1}$, where all elements are expressed as vector in the basis $\mathcal{B}$.
Compute $a_{r}=\sigma_{r} \, (\mathtt{mod} \, p)$ where $\, (\mathtt{mod} \, p)$\footnote{$\sigma(\mathtt{mod}p)$ is the rest of the Euclidean division of $\sigma$ by $p$} applies coordinates
by coordinates following the basis $\mathcal{B}$. Take the canonical lift of $a_{r}$ in $\mathbb{Z}[\pi,\beta]$ and compute $m_{r}=\displaystyle\displaystyle\frac{%
\sigma_{r}-a_{r}}{\pi}=\sigma_{r} \,(\mathtt{div}\, p)$\footnote{$\sigma(\mathtt{div}p)$ is the quotient of the Euclidean division of $\sigma$ by $p$}. The feedback function is $f(s)=(a_{1},\cdots ,a_{r-1},a_{r},m_{r})$ and the output function is $g(x_{0},x_{1},\cdots ,x_{r-1},z)=x_{0}$.
The $d$-VFCSR generates an infinite output sequence $\underline{a}%
=(g(s),g(f(s)),g(f^{2}(s)),\cdots )=(a_{0},a_{1},a_{2},\cdots )$. The state $s$ is called the initial state, $q_{1},\cdots ,q_{r}$ are called the
coefficients of the recurrence and the infinite sequence $(m_{r-1},m_{r},\cdots )$ is called the memory sequence.
\end{definition}
Note that we choose connection coefficients not necessarily positive, because if we consider the inverse problem: what is the $d$-VFCSR outputs a given sequence? We remark that the sequence corresponds to a "rationnal" $\displaystyle\frac{s}{q}\in\mathbb{Q}(\pi,\beta)$ where $q\cong -1 \,\mathtt{mod}\,\pi$. The denominator $q$ constructs the $d$-VFCSR whose output is the given sequence (see Definition \eqref{D2}). So we have: 
$$\begin{array}{ll}
q&=-1+\underset{i=0}{\overset{i=d-1}{\sum}}\underset{j=0}{\overset{j=n-1}{\sum}}\delta_{i,j}\pi^{i}\beta^{j}
\textrm{ with $\delta_{i,j}\in\mathbb{Z}$ for $i\neq 0$ and $\delta_{0,j}\in p.\mathbb{Z}$}\\
q&=-1+\underset{i=0}{\overset{i=d-1}{\sum}}\underset{k=0}{\overset{k=r_{i,j}}{\sum}}\underset{j=0}{\overset{j=n-1}{\sum}}
\mathtt{sgn}(\delta_{i,j})\delta^{i,j}_{k}\beta^{j}\pi^{dk+i}\textrm{ with $\delta_{k}^{i,j}\in \left\{0,1,\cdots,p-1\right\}$} \end{array}$$
and $\delta_{0}^{0,j}=0$. We set $s=\max\left\{r_{i,j}\right\}$, $r=ds+d-1$ and $q^{dk+i}_{j}=\mathtt{sgn}(\delta_{i,j})\delta^{i,j}_{k}$ or $0$ for $k>r_{i,j}$, then:
$$\begin{array}{rrcl}
q&=-1+\underset{I=1}{\overset{I=r}{\sum}}&\underset{j=0}{\overset{j=n-1}{\sum}}
q^{I}_{j}\beta^{j}&\pi^{I}\textrm{ with $q_{j}^{I}\in S$},\\
q&=-1+\underset{I=1}{\overset{I=r}{\sum}}&q_{I}&\pi^{I}\textrm{ with $q_{I}\in S[\beta]$}.
\end{array}$$
The integer $r$ is the size of the $d$-VFCSR and the collection of $q_{I}$ are the connection coefficients. Note also that the connection integer does not give an unique $d$-VFCSR. In fact, we take a simple example to show this: fix $d=p=2$, $n=1$; for $q=-1+\pi=-1-(p-1)\pi+\pi^{3}$, the connection integer gives at least two $d$-VFCSR, one of size $1$ and an other of size $3$.

 Finally, the $p$-ary FCSRs correspond to $d$-VFCSRs with $d=1$ and $n=1$. The VFCSRs correspond to $d$-VFCSR with $d=1$ and $n>1$. The $p$-ary $d$-FCSRs correspond to $d$-VFCSR with $d>1$ and $n=1$. So $d$-VFCSRs generalize FCSRs, $d$-FCSRs and VFCSRs. The figure \eqref{d-VFCSR-Q} represents a $d$-VFCSR in the case $d=2$ and $n=2$.
\subsection{Vector calculus of $d$-VFCSR over $(\mathbb{F}_{p},P,\mathcal{B})$}
We write all elements or its lift in the vectorial basis $\mathcal{B}$. For all 
$$
\begin{array}{ll}
i\in \mathbb{N},\text{ }a_{i}=\sum\limits_{j=0}^{j=n-1}a_{j}^{i}\beta^{j};%
\text{ $a_{j}^{i}\in \mathbb{F}_{p}$ \ }\\
i\geq r-1,\text{\ }m_{i}=\sum\limits_{k=0}^{k=d-1}\sum\limits_{j=0}^{j=n-1}m_{k,j}^{i}\pi^{k}\beta^{j};\text{
$m_{k,j}^{i}\in \mathbb{Z}$ }\\ 
1\leq i\leq r,\text{\ }%
q_{i}=\sum\limits_{j=0}^{j=n-1}q_{j}^{i}\beta^{j};\text{ $q_{j}^{i}\in 
S$,}\\
i\geq r,\text{ }\sigma
_{i}=\sum\limits_{k=0}^{k=d-1}\sum\limits_{j=0}^{j=n-1}\sigma _{k,j}^{i}\pi^{k}\beta^{j};\text{ $\sigma
_{k,j}^{i}\in \mathbb{Z}$.}%
\end{array}
$$
In calculating $\sigma$, we find a polynomial expression of
degree $2n-2$ in $\beta$ thus we must eliminate the degree greater than $n-1$ to obtain
the coordinates for $\sigma $ in the basis $\mathcal{B}$. For this, we must
express the power of $\beta$ in terms of $\mathcal{B}$ with the polynomial 
$P$. So we set for all $j\geq n$, $\beta^{j}=\sum\limits_{t=0}^{t=n-1}b_{t}^{j}\beta^{t}$, where  $b_{t}^{j}\in \mathbb{Z}$. We get the coordinates of $\sigma$: for $z\geq r,$ 
$$
\begin{array}{rl}
\sigma_{z}=&\sum\limits_{t=0}^{t=n-1}\Big[\sum\limits_{k=0}^{k=t}\sum%
\limits_{i=1}^{i=r}(q_{k}^{i}a_{t-k}^{z-i})+\sum%
\limits_{j=n}^{j=2n-2}(b_{t}^{j}\sum\limits_{k=j-n+1}^{k=n-1}\sum\limits_{i=1}^{i=r}(q_{k}^{i}a_{j-k}^{z-i}))
+\sum\limits_{i=0}^{i=d-1}m_{i,t}^{z-1}\pi^{i}\Big]\beta^{t}\\
\sigma^{z}_{t}=&\underbrace{\underset{k=0}{\overset{k=t}{\sum}}\sum%
\limits_{i=1}^{i=r}(q_{k}^{i}a_{t-k}^{z-i})+\sum%
\limits_{j=n}^{j=2n-2}(b_{t}^{j}\sum\limits_{k=j-n+1}^{k=n-1}\sum\limits_{i=1}^{i=r}(q_{k}^{i}a_{j-k}^{z-i}))+m_{0,t}^{z-1}}_{\displaystyle\sigma^{z}_{0,t}}+\sum\limits_{i=1}^{i=d-1}\underbrace{m_{i,t}^{z-1}}_{\displaystyle\sigma^{z}_{i,t}}\pi^{i},
\end{array}
$$
the coordinates of the output sequence: 
\begin{equation}
\begin{array}{rl}
a_{z}=&\sum\limits_{t=0}^{t=n-1}\Big[\sum\limits_{k=0}^{k=t}\sum%
\limits_{i=1}^{i=r}(q_{k}^{i}a_{t-k}^{z-i})+\sum%
\limits_{j=n}^{j=2n-2}(b_{t}^{j}\sum\limits_{k=j-n+1}^{k=n-1}\sum\limits_{i=1}^{i=r}(q_{k}^{i}a_{j-k}^{z-i}))+m_{0,t}^{z-1}\Big]\, \mathtt{mod}\,\pi\,(\beta\,\mathtt{mod}\, \pi)^{t}\\
a_{t}^{z}=&\sigma^{z}_{0,t}\, \mathtt{mod} \, p,
\end{array}
\label{E4}
\end{equation}
and the coordinates of the memory values in the basis $\mathcal{B}$:
\begin{equation}
\begin{array}{rl}
m_{z}=\displaystyle\frac{\sigma_{z}-a_{z}}{\pi}=&\sum\limits_{i=0}^{i=d-2}\sum\limits_{t=0}^{t=n-1}\sigma_{i+1,t}^{z}\pi^{i}\beta^{t}+\sum\limits_{t=0}^{t=n-1}\displaystyle\frac{\sigma_{0,t}^{z}-\sigma_{0,t}^{z}\, \mathtt{mod}\,p}{2}\pi^{d-1}\beta^{t}\\
m_{z}=&\sum\limits_{i=0}^{i=d-2}\sum\limits_{t=0}^{t=n-1}m_{i+1,t}^{z-1}\pi^{i}\beta^{t}+\sum\limits_{t=0}^{t=n-1}\displaystyle\frac{\sigma_{0,t}^{z}-\sigma_{0,t}^{z}\, \mathtt{mod}\,p}{2}\pi^{d-1}\beta^{t}\\
m^{z}_{t}=&\sum\limits_{i=0}^{i=d-2}\underbrace{m_{i+1,t}^{z-1}}_{\displaystyle m_{i,t}^{z}}\pi^{i}+\underbrace{\frac{\sigma_{0,t}^{z}-\sigma_{0,t}^{z}\, \mathtt{mod}\,p}{p}}_{\displaystyle m_{d-1,t}^{z}}\pi^{d-1}.
\end{array}
\label{E5}
\end{equation}
\begin{equation}
\begin{array}{rl}a^{z}_{t}&=\sigma_{t}^{z}-\pi m_{t}^{z}=\sigma^{z}_{0,t}
+\sum\limits_{i=1}^{i=d-1}m_{i,t}^{z-1}\pi^{i}-
\sum\limits_{i=0}^{i=d-2}m_{i+1,t}^{z-1}\pi^{i+1}
- m_{d-1,t}^{z}p\\
&=\displaystyle\sigma^{z}_{0,t}- m_{d-1,t}^{z}p=\displaystyle\sigma^{z}_{0,t}- m_{0,t}^{z+d-1}p.
\end{array}
\label{E5-2}
\end{equation}
\section{Norm and Analysis of $d$-VFCSRs}
The vectorial output sequence $\underline{a}$ corresponds to $n$ binary sequences $\underline{a}_{j}=(a_{j}^{i})_{i\in \mathbb{N}}$. We consider $\underline{a}_{k,j}$ the $d$-decimations cyclic shift of $\underline{a}_{j}$ for each $0\leq j\leq n-1$:
\begin{equation}
\underline{a}=(a_{i})_{i\in \mathbb{N}}=\sum_{j=0}^{j=n-1}(a_{j}^{i})_{i\in 
\mathbb{N}}\overline{\beta}^{j}=\sum_{j=0}^{j=n-1}\underline{a}_{j}\overline{\beta}^{j},
\label{E6} 
\end{equation}
and we have for all $0\leq k\leq d-1$,
\begin{equation}
\underline{a}_{k,j}=(a_{j}^{di+k})_{i\in 
\mathbb{N}}.
\label{E7} 
\end{equation}
With this vectorial vision of elements, we construct a similar
correspondence to that of VFCSR. For each $\underline{a}_{j}$, we associate its $\pi$-adic development and we obtain a $\pi$-adic vector $\alpha=(\alpha_{j})_{0\leq j\leq n-1}$:
\begin{equation}
{(\mathbb{F}_{p^{n}})}^{\mathbb{N}} \rightarrow  {(\mathbb{Z}_{p}[\pi])}^{n},\, 
\underline{a}  \mapsto  \alpha=\Big(\displaystyle\sum\limits_{z\in \mathbb{N}}a_{j}^{z}\pi^{z}%
\Big)_{j},
\label{E8} 
\end{equation}
where $\mathbb{Z}_{p}[\pi]$ is the ring of the $\pi$-adic integers. It is a free $\mathbb{Z}_{p}$-module of rank $d$ with the canonical basis $\left\{1,\pi,\ldots,\pi^{d-1}\right\}$. For each $j$, set $\alpha_{j}=\sum_{k=0}^{d-1}(\sum_{z\in \mathbb{N}}a_{j}^{dz+k}p^{z})\pi^{k}$. So, for each $\underline{%
a}_{k,j}$, we associate its $p$-adic development and we obtain a $p$-adic vector $%
\alpha^{'} =(\alpha_{k,j})_{k,j}$ associated to $\underline{a}$ 
\begin{equation}
{(\mathbb{F}_{p^{n}})}^{\mathbb{N}}  \rightarrow  {(\mathbb{Z}_{p})}^{dn} ,\, 
\underline{a}  \mapsto  \alpha^{'}=\Big(\displaystyle\sum\limits_{z\in \mathbb{N}}a_{j}^{dz+k}p^{z}%
\Big)_{k,j}.
\label{E9}
\end{equation}
\begin{definition}
\label{D2} We define the connection integer of the $d$-VFCSR by $q=-1+\overset{i=r}{\underset{i=1}{\sum}}q_{i}\pi^{i}\in\mathbb{Z}[\pi,\beta]$. We set for each $0\leq j\leq n-1$, $\tilde{q}_{j}=\overset{i=r}{\underset{i=1}{\sum}}q_{j}^{i}\pi^{i}$. For each $0\leq k\leq d-1$, take $\tilde{q}_{k,j}=\overset{}{\underset{i;1\leq di+k\leq r}{\sum}} q_{j}^{di+k}p^{i}$. We have
\begin{equation}
q=-1+\displaystyle\sum\limits_{j=0}^{j=n-1}\tilde{q}_{j}\beta^{j}-1=
\displaystyle\sum\limits_{j=0}^{j=n-1}\displaystyle\sum\limits_{k=0}^{k=d-1}
\tilde{q}_{k,j}\pi^{k}\beta^{j}.
\label{E10} 
\end{equation}
\end{definition}
Using \eqref{E4}, \eqref{E5}, \eqref{E5-2}, \eqref{E7}, \eqref{E9} and \eqref{E10}, $\alpha^{'}$ verifies the following integral linear system 
\begin{equation}
\left\{\begin{array}{l}
 \alpha_{k,j}-\displaystyle\sum\limits_{l=0}^{l=j}\displaystyle\sum\limits_{u=0}^{u=k}\tilde{q}_{u,l}\alpha_{k-u,j-l}-\displaystyle\sum\limits_{l=0}^{l=j}\displaystyle\sum\limits_{u=k+1}^{u=d-1}\tilde{q}_{u,l}\alpha_{k-u+d,j-l}\,p\\
 -\displaystyle\sum\limits_{t=n}^{t=2n-2}b_{j}^{t}\displaystyle\sum\limits_{l=t+1-n}^{l=n-1}\displaystyle\sum\limits_{u=0}^{u=k}\tilde{q}_{u,l}\alpha_{k-u,t-l}\\
 -\displaystyle\sum\limits_{t=n}^{t=2n-2}b_{j}^{t}\displaystyle\sum\limits_{l=t+1-n}^{l=n-1}\displaystyle\sum\limits_{u=k+1}^{u=d-1}\tilde{q}_{u,l}\alpha_{k+d-u,t-l}\, p=\tilde{p}_{k,j}
\end{array}\right\}_{0\leq k\leq d-1,0\leq j\leq n-1}.
\label{E11}
\end{equation}%
This system can be written in matrix form with a square matrix $M^{'}\in \mathcal{M}_{nd}(\mathbb{Z})$. This matrix 
is called \textit{the connection integer matrix of the $d$-FCSR over }$(\mathbb{F}_{p},P,\mathcal{B})$. Furthermore, using \eqref{E4}, \eqref{E5}, \eqref{E5-2}, \eqref{E6}, \eqref{E8} and \eqref{E10}, $\alpha$ verifies the following integral linear system 
\begin{equation}
\left\{
\alpha_{j}-\displaystyle\sum\limits_{l=0}^{l=j} \tilde{q}_{l}\alpha_{j-l}-\displaystyle\sum\limits_{t=n}^{t=2n-2}b_{j}^{t}\displaystyle\sum\limits_{l=t+1-n}^{l=n-1} \tilde{q}_{l}\alpha_{t-l}=\tilde{p}_{j}
\right\}_{0\leq j\leq n-1}.
\label{E12}
\end{equation}%
This system can be written in matrix form with a square matrix $M\in \mathcal{M}_{n}(\mathbb{Z}[\pi])$. We have the following theorem:
\begin{theorem}\label{thm1} Consider a $d$-VFCSR over $(\mathbb{F}_{p},P,\mathcal{B})$ of length $r$ with connection coefficient $q_{1},\ldots,q_{r}$. Let $M^{'}$ the square matrix defined by \eqref{E11} and $M$ the square matrix defined by \eqref{E12}. For all output sequence $\underline{a}$, consider $\alpha$ the $\pi$-adic vector defined by \eqref{E8} and $\alpha^{'}$ the $p$-adic vector defined by \eqref{E9}. We have the following assertions: \\
\linebreak
1. $M.\alpha=y$ where $y=(\tilde{p}_{j})_{j}\in (\mathbb{Z}[\pi])^{n}$ is defined by \eqref{E12}.\\
2. $M^{'}.\alpha^{'}=y^{'}$ where $y^{'}=(\tilde{p}_{k,j})_{k,j}\in (\mathbb{Z})^{nd}$ is defined by \eqref{E12}.\\
3. The matrix $M^{'}$ is the matrix in the canonical basis $\mathcal{B}$ of the $\mathbb{Q}$-linear transformation defined as the multiplication by $-q$.  The matrix $M$ is the matrix in the canonical basis $\left\{ 1,\beta,\cdots,\beta^{n-1}\right\}$ of the $\mathbb{Q}(\pi)$-linear transformation defined as the multiplication by $-q$. \\
4. We have $\det(M)=\mathtt{N}^{\mathbb{Q}(\pi,\beta)}_{\mathbb{Q}(\pi)}(-q)$, $\det(M^{'})=\mathtt{N}^{\mathbb{Q}(\pi,\beta)}_{\mathbb{Q}}(-q)$ and $\det(M^{'})=\mathtt{N}^{\mathbb{Q}(\pi)}_{\mathbb{Q}}(\det(M))$.\\
5. The coefficients of the diagonal of $M$ are congruent to 1 modulo $\pi$ and the other coefficients are multiples of $\pi$.\\
6. $\det(M)$ is congruent to 1 modulo $\pi$ and $M$ is invertible in $\mathcal{M}_{n}(\mathbb{Z}[\pi])$.\\
7. $\det(M^{'})$ is congruent to 1 modulo $p$ and $M^{'}$ is invertible $\mathcal{M}_{nd}(\mathbb{Z})$.\\
8. The $\pi$-adic vector $\alpha$ is in $\displaystyle\frac{1}{\det (M)}\mathbb{Z}[\pi]^{n}$ and the $p$-adic vector $\alpha^{'}$ is in $\displaystyle\frac{1}{\det (M^{'})}\mathbb{Z}^{nd}$.
\end{theorem}
\begin{proof}
- The points (1) and (2) are direct consequences of the calculation of $\alpha$ and $\alpha^{'}$ using recursive relations (\eqref{E4}, \eqref{E5}, \eqref{E5-2}, \eqref{E7}, \eqref{E9} and \eqref{E10}) between the connection coefficients, the initial state and the initial memory. To find the matrix $M$, we use the same calculations considering all elements as vectors on $\mathbb{Q}(\pi)$.  \\
- For the point (3), if we read the linear transformation defined by $-q$ as a vector over $\mathbb{Z}$ and over $\mathbb{Z}[\pi]$, we find respectively the matrix $M^{'}$ and $M$. \\
- The point (4) is direct consequence of the definition of the norm and its basic properties.\\
- For the points (5), $M$ is equal to the identity matrix minus a matrix whose coefficients are linear combinations of the collection $(\tilde{q}_{0},\cdots,\tilde{q}_{n-1})$ and the $\tilde{q}_{j}$ are multiples of $\pi$. The point $(6)$ is a direct consequence of $(5)$. \\
- The norm $\mathtt{N}_{\mathbb{Q}}^{\mathbb{Q}(\pi)}(\det(M))$ is the determinant of the matrix \eqref{E} with $x_{i}$ is the $i$-th component of $\det(M)$ with respect to the basis $\left\{1,\ldots,\pi^{d-1}\right\}$. By (6), $x_{0}$ is congruent to 1 modulo $p$. So the determinant of this matrix is congruent to 1 modulo $p$ and $M^{'}$ is invertible. The point $(7)$ is showed.\\
- For the point (8), $M^{'}$ and $M$ have a comatrix respectively in $\mathbb{Z}$ and $\mathbb{Z}[\pi]$, and we have:
$$
\begin{array}{l}
\alpha^{'} =\displaystyle\frac{1}{|\det (M^{'})|}\mathtt{sgn}(\det M^{'})\mathtt{Comat}(M^{'}).y^{'}\textrm{ and }\\
\alpha= \displaystyle\frac{1}{|\det (M)|}\mathtt{sgn}(\det M)\mathtt{Comat}(M).y. 
\end{array}
$$
\end{proof}
For all $q\in \mathbb{Z}[\pi,\beta]$ such that $q\cong -1 \,\mathtt{mod}\, \pi$, we denote $\mathtt{N}^{\mathbb{Q}(\pi,\beta)}_{\mathbb{Q}(\pi)}(-q)$ by $N$ and $\mathtt{N}^{\mathbb{Q}(\pi,\beta)}_{\mathbb{Q}}(-q)$ by $N^{'}$. $N$ is an element of $\mathbb{Z}[\pi]$ represented by an $n$-form defined by $M$ with arguments $(\displaystyle\tilde{q}_{0}-1,\tilde{q}_{1},\ldots,\tilde{q}_{n-1})$. As far as, $N^{'}$ is an integer represented by an $nd$-form defined by $M^{'}$ with arguments $(\tilde{q}_{0,0}-1,\tilde{q}_{1,0},\ldots,\tilde{q}_{d-1,n-1})$.
\begin{figure}[]
\begin{center}
\includegraphics[height=9cm,width=11cm]{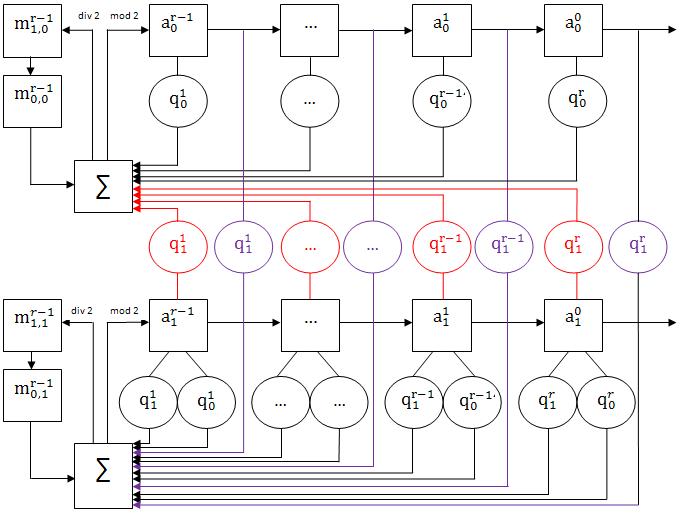}
\caption{Representation of $d$-VFCSR with $n=2$ and $d=2$.}
\label{d-VFCSR-Q}
\end{center}
\end{figure}
\section{Periodicity}
In this section, we discuss the periodicity of output sequences of $d$-VFCSRs.
\begin{theorem}\label{thm2}
Consider a $d$-VFCSR over $(\mathbb{F}_{p},P,\mathcal{B})$. Let an output vectorial sequence $\underline{a}$.\\
1. For all $0\leq j\leq n-1$ and for all $0\leq k\leq d-1$, $\underline{a}_{k,j}$ is eventually periodic and its period divides $\mathtt{ord}_{N^{'}}(p)$, where $\mathtt{ord}_{N^{'}}(p)$ is the order of $p$ modulo ${N}^{'}$.\\
2. For all $0\leq j\leq n-1$, $\underline{a}_{j}$ is eventually periodic and its period divides $d.\mathtt{ord}_{N^{'}}(p)$.\\
3. $\underline{a}$ is eventually periodic and its period is the lcm of the periods of the collection $\underline{a}_{j}$ and divides $d.\ord_{N^{'}}(p)$.
\end{theorem}
\begin{proof}
- The second point of Theorem \eqref{thm1} shows that $\alpha_{k,j}=\frac{y_{k,j}}{N^{'}}$. A $p$-adic integer is a rational if and only if its subsequence is eventually periodic, then the sequence $\underline{a}_{k,j}$ is eventually periodic and its period is equal to the order of $p$ modulo $\frac{N^{'}}{\gcd(N^{'},y^{'}_{k,j})}$. So the period divides the order of $p$ modulo the denominator.\\
- For all $j$, $\underline{a}_{j}=(a_{j}^{0},a_{j}^{1},\ldots,a_{j}^{i},\ldots)$. For all $i$, it exists a pair $(I,k)$ such that $i=dI+k$ with $0\leq k\leq d-1$. We have $a_{j}^{i+d\ord_{N^{'}}(p)}=a_{j}^{dI+k+d\ord_{N^{'}}(p)}=a_{j}^{d(I+\ord_{N^{'}}(p))+k}=a_{j}^{dI+k}=a_{j}^{i}$. So the period of $\underline{a}_{j}$ divides $d.\ord_{N^{'}}(p)$.\\
- The same proof as in [\cite{marjane2010},page 245,Prop 2].
\end{proof}
\begin{corollary}\label{cor1}If $N^{'}$ is a prime number, $p$ is a primitive root modulo $N^{'}$, $d$ is relatively with $N^{'}-1$, then the period of $\underline{a}$ is equal to $N^{'}-1$.
\end{corollary}
\begin{proof}
If $N^{'}$ is prime and $p$ is a primitive root modulo $N^{'}$, then $\ord_{N^{'}}(p)=N^{'}-1$. By the theorem \eqref{thm2}, the period of $\underline{a}_{k,j}$ is $1$ or $N^{'}-1$ for all $k$ and $j$. If for all $k$ and $j$, the period of $\underline{a}_{k,j}$ is $1$, then the period of $\underline{a}_{j}$ is $1$ and the period of $\underline{a}$ is $1$. Otherwise there exist $k$ and $j_{0}$ such that the period $\underline{a}_{k,j_{0}}$ is $N^{'}-1$. If $d$ is relatively with $N^{'}-1$, the application $$\left\{0,d,\ldots,d(N^{'}-2)\right\}\rightarrow \left\{0,1,\ldots,N^{'}-2\right\},\quad dx\mapsto (dx) \mathtt{mod}(N^{'}-1)$$ is injective. In fact $(dx) \mathtt{mod}(N^{'}-1)=(dy) \mathtt{mod}(N^{'}-1)\Leftrightarrow N^{'}-1 \mid d(x-y)\Rightarrow N^{'}-1 \mid (x-y)$. Since $|x-y|<N^{'}-1$, then $x-y=0$. The both set have the same cardinal, then the application is injective. \\
So for $j_{0}$, $a_{j_{0}}^{i+k+N^{'}-1}=a_{j_{0}}^{d(d^{-1}i)+k+N^{'}-1}=a_{j_{0}}^{d(d^{-1}i)+k}=a_{j_{0}}^{i+k}$. Then the period of $\underline{a}_{j_{0}}$ divides the period of $\underline{a}_{k,j_{0}}$. By the same argument, the period of $\underline{a}_{k,j_{0}}$ divides the period of $\underline{a}_{j_{0}}$. So we have an equality and the period of $\underline{a}$ is $\mathtt{lcm}(1,N^{'}-1)=N^{'}-1$.
\end{proof}
\begin{rmq}\ \\
- If we have only $N^{'}$ is a prime number and $p$ is a primitive root modulo $N^{'}$, then the period of $\underline{a}_{k,j}$ is $N^{'}-1$ or 1. If there exists $k$ and $j_{0}$ such that the period is $N^{'}-1$, i.e. $\underline{a}_{j}$ is not trivial, then $N^{'}-1$ divides the period of $\underline{a}_{j}$. By the theorem \eqref{thm2}, the period of $\underline{a}$ is a multiple of $N^{'}-1$.\\
- If we have only $N^{'}$ is a prime number, then the period of $\underline{a}_{k,j}$ is $\ord_{N^{'}}(p)$ or 1. If there exists $k$ and $j_{0}$ such that the period is $\ord_{N^{'}}(p)$, i.e. $\underline{a}_{j}$ is not trivial, then $\ord_{N^{'}}(p)$ divides the period of $\underline{a}_{j}$. By the theorem \eqref{thm2}, the period of $\underline{a}$ is a multiple of $\ord_{N^{'}}(p)$.
\end{rmq}
\section{First case: $p=2$, $n=2$ and $d=2$}
We place ourselves in the binary case $p=2$. In this part, we will present the first case $n=2$ and $d=2$. The case $n=2$ is special because to build our
vectorial $d$-FCSR, there is a single irreducible polynomial of degree 2 on $\mathbb{F}_{2}$: $X^{2}-X-1$ modulo 2. We take $\pi=2^{\frac{1}{2}}$ and $\beta$ such that $\beta^{2}=\beta+1$. We analyse the quadratic case of VFCSRs. From
Equations \eqref{E4} and \eqref{E5}, operations are defined as follows:
\begin{enumerate}
\item Form integers $\sigma _{0,0}^{z}$ and $\sigma_{0}^{z}$, as follows $\sigma_{0,0}^{z}=\underset{i=1}{\overset{i=r}{\sum}}\left(
q_{1}^{i}a_{1}^{z-i}+q_{0}^{i}a_{0}^{z-i}\right)+m_{0,0}^{z-1}$, $\sigma_{1,0}^{z}= m_{1,0}^{z-1}$, $\sigma_{0,1}^{z}=\underset{i=1}{\overset{i=r}{\sum}}\left(
q_{1}^{i}a_{1}^{z-i}+q_{1}^{i}a_{0}^{z-i}+q_{0}^{i}a_{1}^{z-i}\right)
+m_{0,1}^{z-1}$ and $\sigma_{1,1}^{z}= m_{1,1}^{z-1}$.

\item Shift the content of the main register register on step to the right and shift the content of the carry register to the bottom, while outputing the rightmost bits $%
a_{0}^{z-i}$ and $a_{1}^{z-i}$ and the lowermost carry $m_{0,1}^{z-1}$ and $m_{0,0}^{z-1}$ as shown in Fig.\eqref{d-VFCSR-Q},

\item Put and replace $a_{0}^{z}=\sigma _{0,0}^{z}\,(\mathtt{mod}\, 2)$, $a_{1}^{z}=\sigma_{0,1}^{z}\,(\mathtt{mod}\, 2)$, $m_{1,0}^{z}=\sigma _{0,0}^{z}\,(\mathtt{div}\, 2)=\frac{\sigma
_{0,0}^{z}-a_{0}^{z}}{2}$ and $m_{1,1}^{z}=\sigma_{0,1}^{z}\,(\mathtt{div}\, 2)=\frac{\sigma _{0,1}^{z}-a_{1}^{z}}{2}$.
\end{enumerate}
The matrix $M$ is of the form 
$\left(\begin{array}{cc}
1-\tilde{q}_{0}&-\tilde{q}_{1}\\
-\tilde{q}_{1}&1-\tilde{q}_{0}-\tilde{q}_{1}\\
\end{array}\right)$ 
and the matrix $M^{'}$ is of the form:
\begin{equation}
\left(\begin{array}{cc|cc}
1-\tilde{q}_{0,0}&-2\tilde{q}_{1,0}&-\tilde{q}_{0,1}&-2\tilde{q}_{1,1}\\
-\tilde{q}_{1,0}&1-\tilde{q}_{0,0}&-\tilde{q}_{1,1}&-\tilde{q}_{0,1}\\
\hline
-\tilde{q}_{0,1}&-2\tilde{q}_{1,1}&1-\tilde{q}_{0,0}-\tilde{q}_{0,1}&-2\tilde{q}_{1,0}-2\tilde{q}_{1,1}\\
-\tilde{q}_{1,1}&-\tilde{q}_{0,1}&-\tilde{q}_{1,0}-\tilde{q}_{1,1}&1-\tilde{q}_{0,0}-\tilde{q}_{0,1}\\
\end{array}\right).
\end{equation} 
From the definition of $\tilde{q}_{k,j}$, we have 
\begin{equation}
\begin{array}{l}
x=\tilde{q}_{0,0}=\sum\limits_{1\leq 2i\leq r}q_{0}^{2i}2^{i}=\sum\limits_{1\leq i\leq [r/2]}q_{0}^{2i}2^{i},\\
y=\tilde{q}_{1,0}=\sum\limits_{1\leq 2i+1\leq r}q_{0}^{2i+1}2^{i}=\sum\limits_{0\leq i\leq [(r-1)/2]}q_{0}^{2i+1}2^{i},\\
z=\tilde{q}_{0,1}=\sum\limits_{1\leq 2i\leq r}q_{1}^{2i}2^{i}=\sum\limits_{1\leq i\leq [r/2]}q_{1}^{2i}2^{i}\quad and\\
t=\tilde{q}_{1,1}=\sum\limits_{1\leq 2i+1\leq r}q_{1}^{2i+1}2^{i}=\sum\limits_{0\leq i\leq [(r-1)/2]}q_{1}^{2i+1}2^{i}.
\end{array}
\end{equation}
The coefficients $\tilde{q}_{0,0}$ and $\tilde{q}_{0,1}$ are even and $\tilde{q}_{1,0}$ and $\tilde{q}_{1,1}$ can be odd. We must research prime numbers represented by $N^{'}$ with these arguments respecting these conditions. With a simple program of matlab, we generate the numbers represented by this $4$-form (see Table \eqref{quadruplet1}).\\
\textbf{Example 1}: If we take $N^{'}=151$ with the quadruplet $(0,1,2,2)$, then we find the following connection integer: $q=q_{1}\pi+q_{2}\pi^{2}+q_{3}\pi^{3}-1=\pi+\beta\pi^{2}+\beta\pi^{3}-1=\pi+2\beta+2\pi\beta-1.$. The quadruplet $(0,1,2,2)$ corresponds to the $d$-VFCSR of size 3 represented in the figure \eqref{d-VFCSR-151} and generates sequences with the period is a multiple of $\ord_{151}(2)=15$. For example, the initial state $(a_{0},a_{1},a_{2},m)=(1,1+\beta,\beta,5-\beta+4\pi\beta)$ generates the eventually periodic sequences in Table \eqref{sequence150} with the period is equal to 15.\\
\textbf{Example 2}: We recover also the example of Klapper and Xu \cite{klapper1999} with $q=\pi+(1+\beta)\pi^{3}-1=3\pi+2\pi\beta-1$ represented here by $N^{'}=401$ and the quadruplet $(0,3,0,2)$. The quadruplet $(0,3,0,2)$ corresponds to the $d$-VFCSR of size 3 represented in the figure \eqref{d-VFCSR-401} and generates sequences with the period is a multiple of $\ord_{401}(2)=200$.\\
\textbf{Example 3}: If we take $N^{'}=409$ with the quadruplet $(0,1,4,3)$, then we find $q=q_{1}\pi+q_{2}\pi^{2}+q_{3}\pi^{3}+q_{4}\pi^{4}-1=(1+\beta)\pi+\beta\pi^{3}+\beta\pi^{4}-1=\pi+7\beta\pi-1.$. This quadruplet corresponds to the $d$-VFCSR of size $4$ in the figure \eqref{d-VFCSR-409} and generates sequences with the period is $\ord_{409}(2)=204$ or $408$. For example, the initial state $(a_{0},a_{1},a_{2},a_{3},m)=(1,1+\beta,1+\beta,\beta,5-\beta+4\pi\beta)$ generates an eventually periodic sequences with the period is 408.
\begin{table}[]
\begin{center}
\tabcolsep=2pt
\begin{tabular}{|c|c|c|c|c||c|c|c|c|c||c|c|c|c|c||c|c|c|c|c|}
\hline
$N^{'}$&$x$&$y$&$z$&$t$&$N^{'}$&$x$&$y$&$z$&$t$&$N^{'}$&$x$&$y$&$z$&$t$&$N^{'}$&$x$&$y$&$z$&$t$\\
\hline
9&2&1&4&0& 
25&6&2&0&1& 
41&2&1&4&2& 
31&2&1&0&3\\
49&4&1&2&3&
71&2&1&2&3&
79&2&1&4&1&
81&2&2&4&3\\
89&2&3&4&2&
121&2&3&4&1&
151&0&1&2&2&
169&4&2&2&3\\
191&4&3&0&2&
239&4&1&0&2&
241&6&3&0&1&
271&2&3&2&3\\
279&2&3&2&1&
281&6&1&6&2&
289&0&3&0&0&
311&2&2&6&1\\
359&4&0&6&2&
361&6&0&6&0&
369&2&2&6&2&
401&0&3&0&2\\
409&0&1&4&3&
431&0&3&0&1&
439&0&0&6&3&
441&0&2&2&3\\
449&4&1&0&3&
479&4&3&6&2& 
521&4&3&4&1& 
529&6&1&0&0\\ 
569&2&1&6&2& 
601&2&1&6&0&
625&6&0&0&0& 
631&6&3&6&2\\
641&6&3&2&3&
711&2&1&6&1&
719&2&0&6&1&
729&4&3&0&3\\
751&6&1&2&1&
769&0&2&4&3&
761&4&2&0&3&
801&6&1&2&0\\
839&6&0&2&1&
841&2&0&6&0&
881&0&1&6&2&
911&0&2&6&1\\
929&0&3&6&3&
961&0&2&6&0& 
991&6&3&0&2&
1025&6&1&0&3
\\
\hline
\end{tabular} 
\end{center}
\caption{Prime numbers represented by $N^{'}$.}
\label{quadruplet1}
\end{table}
\begin{figure}[]
\begin{center}
\includegraphics[height=7cm,width=9cm]{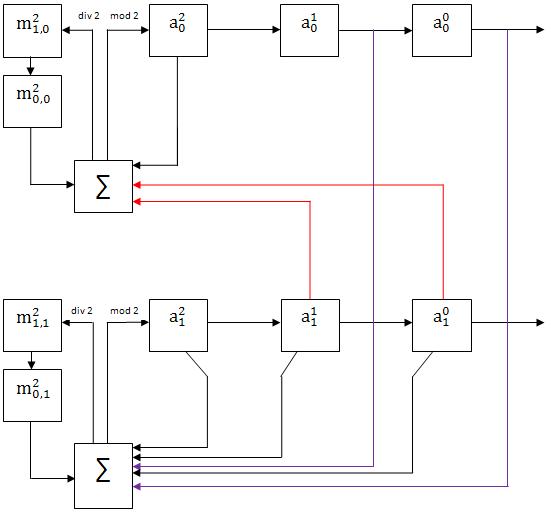}
\caption{Representation of $2$-VFCSR with $n=2$ and $N^{'}=151$.}
\label{d-VFCSR-151}
\end{center}
\end{figure}
\begin{table}[]
\begin{center}
\tabcolsep=2pt
\begin{tabular}{|c|c|c|c|c|c|c|c|c|c|c|c|c|c|c|c|c|c|c|c|c|c|c|c|c|c|c|c|c|c|c|c|c|c|c|c|c|c|}
\hline
$i$&0&&&&&&&&&&11&&&&&&&&&&&&&&&26&&&&&&&&&&\\
\hline
$a_{0}^{i}$&1&1&0&0&0&1&1&0&0&1&0&0&1&0&1&1&1&0&0&1&0&1&0&0&1&0&0&1&0&1&1&1&0&0&1&0\\
$a_{1}^{i}$&0&1&1&1&0&1&0&1&1&0&1&1&1&1&0&1&0&1&1&1&0&1&0&0&1&1&1&1&1&0&1&0&1&1&1&0\\
\hline
$m_{0}^{0}$&5&0&3&1&2&1&2&1&1&2&1&1&2&1&2&1&2&1&1&2&1&2&1&1&1&1&1&2&1&2&1&2&1&1&2&1\\
$m_{0}^{1}$&-1&4&1&4&1&3&1&3&2&2&2&2&2&3&2&3&2&3&2&3&2&3&2&2&1&2&2&2&3&2&3&2&3&2&3&2\\
$m_{1}^{0}$&0&3&1&2&1&2&1&1&2&1&1&2&1&2&1&2&1&1&2&1&2&1&1&1&1&1&2&1&2&1&2&1&1&2&1&2\\
$m_{1}^{1}$&4&1&4&1&3&1&3&2&2&2&2&2&3&2&3&2&3&2&3&2&3&2&2&1&2&2&2&3&2&3&2&3&2&3&2&3\\
\hline
\end{tabular}
\end{center}
\caption{$d$-VFCSR sequence for $N^{'}=151$.}
\label{sequence150}
\end{table}
\begin{figure}[]
\begin{center}
\includegraphics[height=7cm,width=9cm]{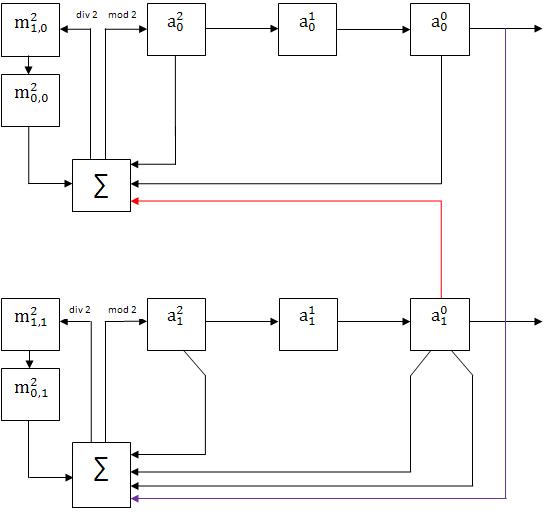}
\caption{Representation of $2$-VFCSR with $n=2$ and $N^{'}=401$.}
\label{d-VFCSR-401}
\end{center}
\end{figure}
\begin{figure}[]
\begin{center}
\includegraphics[height=7cm,width=9cm]{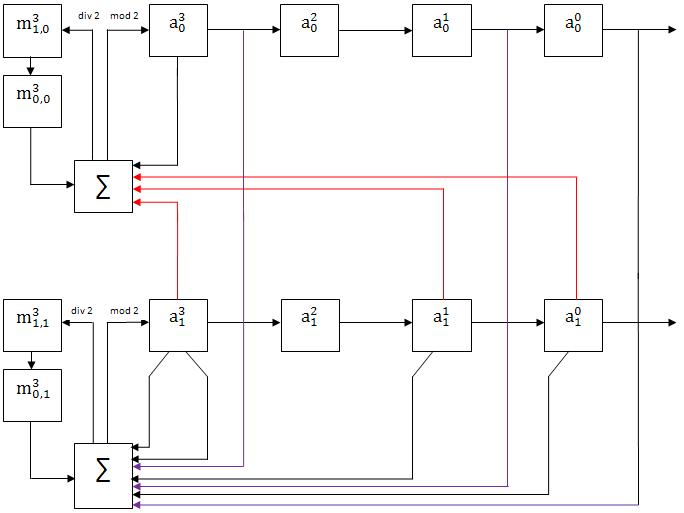}
\caption{Representation of $2$-VFCSR with $n=2$ and $N^{'}=409$.}
\label{d-VFCSR-409}
\end{center}
\end{figure}
\newpage
\section{Conclusion}
In this paper, we have presented a vectorial conception adapted to the $d$-FCSR in order to build $d$-FCSR on any finite fields. These registers generate eventually periodic sequences with the period divides $d.\mathtt{N}_{\mathbb{Q}}^{\mathbb{Q}(\pi,\beta)}(q)$. These registers generalize $p$-ary FCSR, $p$-ary $d$-FCSR and VFCSR and are a particular case of AFSR. In this paper, we give a complete description and a $p$-adic analysis of these registers using vectorial method and illustrating with simple examples. The results are easily implemented in software and hardware. We can easily generate longer sequences, just look for prime numbers represented by the $nd$-forms.   
\section*{Acknowledgments} The author thanks M. Zerzeri and A. Mokrane for their
comments on an earlier version of this paper.
\newpage   
\bibliographystyle{amsalpha}

\end{document}